\documentclass[letterpaper,10pt,conference]{ieeeconf}
 
\IEEEoverridecommandlockouts
\usepackage{amsmath}
\usepackage{amssymb}
\usepackage{bm}
\usepackage{array}
\usepackage{subcaption} 

 \usepackage{stfloats}
\usepackage{graphicx}
\usepackage{booktabs}
\usepackage{cite}
\usepackage{comment}
\usepackage{algorithm}
\usepackage{xcolor}
\usepackage{algpseudocode}

\newtheorem{definition}{Definition}
\newtheorem{theorem}{Theorem}
\newtheorem{lemma}{Lemma}

\title{\LARGE \bf
Timely Persuasion with Opposing Senders: When Does Competition Help the Receiver?
}
 
\author{Ahmet Bugra Gundogan and Melih Bastopcu%
\thanks{A.B.~Gundogan and M.~Bastopcu are with the Department of Electrical and Electronics Engineering, Bilkent University, Ankara, T\"{u}rkiye.
        {\tt\small \{bugra.gundogan, bastopcu\}@bilkent.edu.tr}}%
\thanks{This work was supported in part by the TUBITAK 2232-B program (Project No:124C533).}}
\vspace{-0.3cm}
 
\begin{document}
 
\maketitle
\thispagestyle{empty}
\pagestyle{empty}
\begin{abstract}
Status updates let a receiver track a changing process, but a sender that controls when they are sent can steer the receiver's estimate without misreporting. We study competition between two such senders with opposing preferences, who report on a binary continuous-time Markov chain using state-dependent Poisson rates under individual budgets. The receiver follows the reports through a zero-order-hold estimator only if this yields at least the long-run weighted correct-estimation utility of a prior-based default estimate. We model the interaction as a Stackelberg game in which the sender favored by this default leads, the opposing sender follows, and the receiver then selects its estimator. We characterize the pure-strategy equilibria in closed form. When the follower can obtain positive utility against a silent leader, limited state-$0$ reporting by the leader merely frees part of the follower's budget for state-$1$ reports, so every optimal leader policy is either silence or full-budget state-$0$ reporting, depending on explicit budget thresholds. A single sender with the same total budget leaves the receiver at its default utility, whereas the full-budget intervention outcome makes the receiver's participation constraint slack and strictly improves its utility. At the intervention threshold, silence and full-budget intervention yield identical sender payoffs but different receiver utilities. Numerical results illustrate these regimes and compare equilibrium performance with a receiver-optimal reporting benchmark.
\end{abstract}
\section{Introduction}
\label{sec:introduction}

When a receiver tracks a changing process through status updates,
its estimate is only as current as the latest update. This has
made the age of information a central performance
measure~\cite{Yates12,Survey_AoI}, and keeping information fresh
is the natural goal when the sender wants the receiver to be
accurate. A sender that instead has a stake in the receiver's
estimate can use timing to persuade without misreporting. In other words, by
reporting favorable states promptly and unfavorable ones
less frequently, it lets the receiver's information go stale precisely
when the state turns unfavorable. Since the receiver can ignore
the updates altogether, the sender must still report unfavorable
states often enough to keep the receiver following them \cite{GundoganBastopcu2026Timeliness}. A second
sender with the opposite preference can supply these reports
itself, but doing so also frees the first sender to report
favorable states even more often. This raises two questions: when
does the opposing sender choose to intervene, and when does its
intervention improve the receiver's estimation performance?

We address these questions through a game in which two senders
continuously observe a binary continuous-time Markov chain (CTMC)
and send truthful reports of its state to a common receiver as shown in Fig.~\ref{fig:multisendersystem}. Each
sender's strategy is a pair of state-dependent Poisson reporting
rates, constrained by an individual budget on their sum.
Irrespective of the source state, sender~1 is rewarded whenever the
receiver's estimate is~$1$ and sender~2 whenever it is~$0$, so their
payoffs are constant-sum. The receiver maximizes its long-run
weighted correct-estimation utility by choosing between a
zero-order-hold (ZOH) estimator, which retains the most recently
reported state, and a default constant estimate chosen optimally under the stationary prior; throughout, we consider the regime in which this default estimate is~$0$. The receiver therefore follows the reports only if
ZOH estimation yields at least its default utility, a requirement we
call the participation constraint (PC). Since this default already
favors sender~2, we let it act as the leader and ask when it chooses
to intervene. The resulting Stackelberg game has three stages:
sender~2 commits to its reporting policy; sender~1, the follower,
observes it and commits to its own; and the receiver observes both
and selects its estimator. Both senders then report concurrently
under these fixed policies. The receiver does not use information
conveyed by the absence of reports, and our equilibrium and
receiver-benefit results apply within this estimator architecture.
\begin{figure}[t]
\centering
\includegraphics[width=0.9\columnwidth]{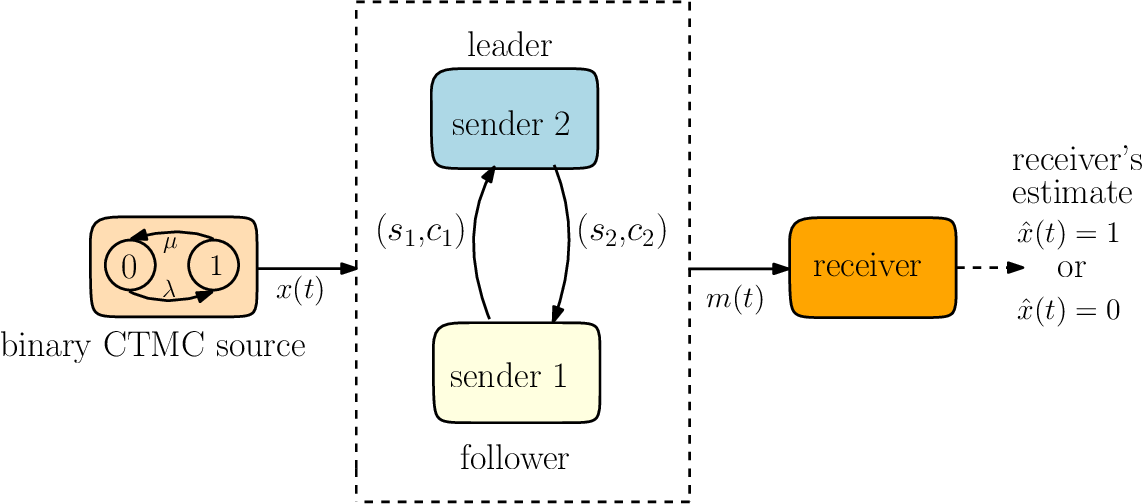}\vspace{-0.15cm}
\caption{System model with two competing senders and one receiver.}
\label{fig:multisendersystem}
\vspace{-0.8cm}
\end{figure}

Within the age-of-information literature, our model is closest to
remote estimation, where sampling policies are designed to improve
the receiver's estimate of Wiener processes~\cite{Sun2020}, binary
infection status~\cite{bastopcu2022using}, and finite-state CTMCs
under sampling rate constraints~\cite{Akar_CTMC}. In our problem,
the reporting rates are chosen by senders with opposing preferences
over this estimate, so estimation performance is an equilibrium
outcome rather than a design objective.

As in Bayesian persuasion~\cite{kamenica2011bayesian} and
information design~\cite{BergemannMorris2019InformationDesign}, our
senders commit to their policies in advance. Dynamic persuasion has
been studied with partially informed receivers~\cite{au2015dynamic},
Markovian states~\cite{renault2017optimal,ashkenazi2023markovtwo},
and stochastic state revelations~\cite{lehrer2025markov}.
Ely~\cite{ely2017beeps} studies dynamic persuasion through messages
about an evolving stochastic process, with a receiver that acts on
its beliefs. Related formulations address sequential disclosure for
change detection~\cite{farhadi2022dynamic}, continuous-time
persuasion through filtering~\cite{aid2025filtering}, and costly
information provision without commitment to future
actions~\cite{che2023dynamic}. In contrast to these models, which allow the sender to shape the receiver's information through signal or disclosure design, our senders cannot alter the content of information and control only the timing of truthful state reports. Our prior
work~\cite{GundoganBastopcu2026Timeliness} studies this form of
timely persuasion with a single sender, and we retain its source,
reporting, and receiver models. There, whenever persuasion is
feasible, the sender reports state~$0$ just often enough to make the
PC bind, so the receiver gains nothing from following the reports.

Asking whether competition can improve on this outcome connects our
work to multi-sender persuasion, where the effect of competition on
information revelation depends on the available signal
structures~\cite{gentzkow2017competition,gentzkow2017multisender}.
Au and Kawai~\cite{au2020competitive} show that increasing the
number of symmetric senders leads to more aggressive disclosure in
their model, and Li and Norman~\cite{LiNorman2021Sequential}
characterize equilibrium outcomes under sequential persuasion. In
the control literature, Velicheti
\emph{et al.}~\cite{Velicheti2025_multi_sender} study simultaneous
sender commitments with Gaussian information and quadratic
objectives, including a dynamic extension, and establish a receiver
benefit from competition relative to cooperation among the senders.
Our senders instead commit sequentially to the rates of truthful
reports. Since $J_1+J_2=1$, maximizing the sum of sender payoffs does
not select a reporting policy. We therefore compare the receiver's
equilibrium utility with that under a single sender with the same
total budget.

We answer both questions posed at the outset. For the first, we
derive the follower's best response and characterize the
pure-strategy Stackelberg equilibria in closed form. When the
follower can obtain positive utility on its own, the leader's
state-$0$ reports first help the follower by freeing its budget for
state-$1$ reports and hurt it only once they alone exceed what
participation requires. The leader therefore chooses between silence
and full-budget state-$0$ reporting. We identify a threshold on the
follower's budget below which silence is optimal regardless of the
leader's budget, and a second threshold governing intervention above
it. For the second, whenever the full-budget intervention outcome is
optimal, the leader's state-$0$ reports make the PC slack and strictly
increase the receiver's utility. In contrast, the optimal
single-sender benchmark leaves the receiver at its default utility. At the intervention threshold, the
two optimal leader policies yield identical sender payoffs but
different receiver utilities. Numerical results illustrate these
regimes and show that an appropriate budget division can attain the
receiver-optimal reporting benchmark for the parameters considered.

\section{System Model and Problem Formulation}
\label{sec:model}

We consider a binary source observed by two senders that report to a
receiver, as shown in Fig.~\ref{fig:multisendersystem}. The source state
$x(t)\in\{0,1\}$ evolves as a continuous-time Markov chain (CTMC), with
transition rate $\lambda>0$ from state~$0$ to state~$1$ and rate $\mu>0$
from state~$1$ to state~$0$. Thus, the source has a unique stationary distribution, given by
\begin{align}\nonumber\\[-21pt]
    \pi_0=\frac{\mu}{\lambda+\mu},\qquad
    \pi_1=\frac{\lambda}{\lambda+\mu}.
    \label{eq:source_stationary}
   \\[-21pt] \nonumber
\end{align}
Both senders continuously observe the source. The receiver does not
observe it directly but knows the transition rates $\lambda$ and $\mu$.

The receiver wants to track the source's state as accurately as possible. To do so,  it forms an estimate $\hat x(t)\in\{0,1\}$ at time $t$ and obtains an instantaneous
utility as
\begin{align}\nonumber\\[-19pt]
    u_R(x,\hat x)
    =q\mathbf 1_{\{x=0,\hat x=0\}}
    +(1-q)\mathbf 1_{\{x=1,\hat x=1\}}, \label{eq:receiver_utility}\\[-19pt]\nonumber
\end{align}
where $q\in(0,1)$ weights correct estimates of the two states and
$\mathbf 1_{\{\cdot\}}$ is the indicator function. Constant estimates
$0$ and $1$ yield average utilities $q\pi_0$ and $(1-q)\pi_1$,
respectively. Throughout the paper, we assume
$q\pi_0>(1-q)\pi_1$, so the receiver's optimal constant estimate is~$0$.
Sender~1 prefers the receiver's estimate to be~$1$, whereas sender~2 prefers it to be~$0$, regardless of the source state:
\begin{equation}
    u_1(x,\hat x)=\mathbf 1_{\{\hat x=1\}},\qquad
    u_2(x,\hat x)=\mathbf 1_{\{\hat x=0\}}.
    \label{eq:sender_utilities}
\end{equation}

Let $s_j$ and $c_j$ denote the reporting rates of sender~$j\in\{1,2\}$ when $x(t)=1$ and $x(t)=0$, respectively. Sender~$j$ selects the stationary reporting policy $\beta_j=(s_j,c_j)$. Conditional on the source trajectory, each sender's reports form a Poisson process with these state-dependent rates, and the two reporting processes are independent. The rates depend only on the current source state, not on past reports or the receiver's estimate. Each report contains the current source state and arrives without delay or error. Following~\cite{GundoganBastopcu2026Timeliness}, we impose the conditional-intensity budget
\begin{equation}
    \beta_j\in\mathcal B_j
    :=\{(s_j,c_j)\in\mathbb R_{\geq0}^{2}:s_j+c_j\leq R_j\},
    \label{eq:reporting_budget}
\end{equation}
where $R_j\geq0$ is sender~$j$'s reporting capacity. With the constraints on the sum of the conditional intensities, the long-run average reporting rate of sender~$j$ is $\pi_1s_j+\pi_0c_j$. The source parameters,
budgets, and utility functions are common knowledge.

Following~\cite[Remark~1]{GundoganBastopcu2026Timeliness}, the receiver chooses between two estimation policies. Under $\sigma_s$, it uses the reports from both senders where each new report replaces the current estimate with the reported source state, which is then held until the next report. Under $\sigma_d$, it ignores all reports and always estimates~$0$. Let $t_k$ denote the $k$th report arrival from either sender. Then, under $\sigma_s$, the receiver's estimate  $\hat x(t)$ becomes
\begin{equation}
\hat x(t)=x(t_k),\qquad t_k\leq t<t_{k+1}.
\label{eq:zoh_estimator}
\end{equation}
Before the first report, the receiver estimates~$0$.

We study a receiver with a fixed hold-last-value interface that can either be enabled or disabled before operation. Accordingly, the receiver selects one of the two admissible policies, $\sigma_s$ or $\sigma_d$, before reporting begins and uses it thereafter. The participation constraint introduced below compares these two modes which does not imply that the held estimate is optimal after every observation history.

We model the interaction as a three-stage Stackelberg game with
sender~2 as the leader and sender~1 as the follower. In the first
stage, sender~2 commits to a reporting policy
$\beta_2\in\mathcal B_2$. In the second stage, sender~1 observes
$\beta_2$ and commits to a reporting policy
$\beta_1\in\mathcal B_1$. In the third stage, the receiver observes
$(\beta_1,\beta_2)$ and selects an estimator
$\sigma\in\{\sigma_s,\sigma_d\}$. These choices remain fixed
throughout the subsequent reporting and estimation. Because the
preferred estimate of sender~2 coincides with the receiver's default
estimate~$0$, the leader effectively decides whether to intervene
against sender~1, which prefers estimate~$1$, anticipating the
subsequent choices of sender~1 and the receiver. The payoff of
player $k\in\{1,2,R\}$ is its long-run average utility
\begin{equation}
    J_k(\beta_1,\beta_2,\sigma)
    =\lim_{T\to\infty}\frac{1}{T}
    \mathbb E\!\left[\int_0^T u_k(x(t),\hat x(t))\,dt\right].
    \label{eq:average_utility}
\end{equation}
Since $u_1(x,\hat x)+u_2(x,\hat x)=1$ for all $(x,\hat x)$, the
senders' payoffs satisfy $J_1+J_2=1$ for every strategy profile
$(\beta_1,\beta_2,\sigma)$. The game between the senders is
therefore constant-sum and strategically equivalent to a zero-sum
game in which sender~2 minimizes the payoff that sender~1 attains
through its best response.

\subsection{Average Utilities and the Participation Constraint}
\label{subsec:utilities}
Under $\sigma_s$, the estimate is updated at the report arrivals of
both senders. Since the two reporting processes are conditionally independent given
the source trajectory, their superposition is conditionally Poisson
with rate $s=s_1+s_2$ in state~$1$ and $c=c_1+c_2$ in state~$0$. Because every report reveals the current source state, the ZOH
estimate depends only on the arrival times of this aggregate process
and not on which sender generated each report. Hence,
$(x(t),\hat x(t))$ is a CTMC on $\{0,1\}^2$ with the same transition
structure as the single-sender model in
\cite[Sec.~II-A]{GundoganBastopcu2026Timeliness}, with $(s,c)$ in
place of that sender's reporting rates. Consequently, the payoffs
under $\sigma_s$ depend on $(\beta_1,\beta_2)$ only through $(s,c)$,
and, with a slight abuse of notation, we write
$J_k(s,c):=J_k(\beta_1,\beta_2,\sigma_s)$ for $k\in\{1,2,R\}$. For
$(s,c)\neq(0,0)$, this chain has a unique stationary distribution,
and evaluating the payoffs under it as in
\cite[Sec.~II-A]{GundoganBastopcu2026Timeliness} gives
\begin{equation}
    J_R(s,c)
    =\frac{q\mu c(\mu+s)+(1-q)\lambda s(\lambda+c)}
    {(\lambda+\mu)(\mu c+\lambda s+cs)},
    \label{eq:receiver_average}
\end{equation}
\begin{equation}
    J_1(s,c)
    =\frac{\lambda s(c+\lambda+\mu)}
    {(\lambda+\mu)(\mu c+\lambda s+cs)},
    \label{eq:sender_average}
\end{equation}
and $J_2(s,c)=1-J_1(s,c)$. When $s=c=0$, no reports arrive and the
estimate remains at its initial value~$0$, so $J_R(0,0)=q\pi_0$,
$J_1(0,0)=0$, and $J_2(0,0)=1$. Under $\sigma_d$, the receiver's estimate is
identically~$0$ for all $t$, so $J_R=q\pi_0$, $J_1=0$, and $J_2=1$ for every
$(\beta_1,\beta_2)$.

Since the receiver obtains $q\pi_0$ under $\sigma_d$ and $J_R(s,c)$ under $\sigma_s$, following the reports is optimal within the receiver's admissible policy set if and only if the participation constraint (PC)
\begin{equation}
    J_R(s,c)\geq q\pi_0
\label{eq:participation}
\end{equation}
holds. We define the term $c_{\min}:=q\mu/(1-q)-\lambda$, which is positive
because $q\pi_0>(1-q)\pi_1$. Subtracting $q\pi_0$ from
\eqref{eq:receiver_average} and using
$(1-q)(\lambda+c)-q\mu=(1-q)(c-c_{\min})$ gives, for
$(s,c)\neq(0,0)$,
\begin{equation}
    J_R(s,c)-q\pi_0
    =\frac{(1-q)\lambda s(c-c_{\min})}
    {(\lambda+\mu)(\mu c+\lambda s+cs)}.
    \label{eq:receiver_gain_identity}
\end{equation}
Hence, for $s>0$, the PC is equivalent to
\begin{equation}
    c_1+c_2\geq c_{\min},
\label{eq:participation_threshold}
\end{equation}
in agreement with \cite[eq.~(10)]{GundoganBastopcu2026Timeliness},
whereas for $s=0$ it holds with equality. When $s>0$, participation therefore requires a minimum aggregate
intensity $c_{\min}$ of state-$0$ reports. In the
single-sender model, this requirement must be met by that sender
alone. Here, it is met jointly by the state-$0$ reports of both
senders.

When the PC holds with equality, the receiver is indifferent between
$\sigma_s$ and $\sigma_d$. We adopt the tie-breaking rule under which
the receiver selects $\sigma_s$ at indifference unless $s=c=0$, in
which case it selects $\sigma_d$. The receiver's response is thus
\begin{equation}
    \!\!\sigma^*(\beta_1,\beta_2)\!=\!
    \begin{cases}
        \sigma_s, &\! J_R(s,c)\!\geq\! q\pi_0\text{ and }(s,c)\!\neq\!(0,0),\\
        \sigma_d, & \text{otherwise},
    \end{cases}\!\!
    \label{eq:receiver_response}
\end{equation}
where $(s,c)$ are the aggregate rates of the senders defined previously.
This rule favors sender~1 and ensures that its maximum is attained
when the PC binds. Under the opposite rule, a follower facing a silent leader with
$R_1>c_{\min}$ would obtain a positive payoff only if
$c_1>c_{\min}$, so its payoff would approach its supremum as $c_1$
approaches $c_{\min}$ from above, without attaining it.

\subsection{Stackelberg Game Formulation}
\label{subsec:game_formulation}

Since sender~1 observes $\beta_2$ and the receiver observes
$(\beta_1,\beta_2)$, their strategies are maps
$\beta_1^*:\mathcal B_2\to\mathcal B_1$ and
$\sigma^*:\mathcal B_1\times\mathcal B_2\to\Sigma$, where
$\Sigma:=\{\sigma_s,\sigma_d\}$.

\begin{definition}
\label{def:equilibrium}
The tuple $\!(\beta_1^*,\beta_2^*,\sigma^*)\!$ is a pure-strategy
Stackelberg equilibrium if, for all $\!\beta_1\!\!\in\!\mathcal B_1$,
$\!\beta_2\!\!\in\!\mathcal B_2$, and $\!\sigma\!\!\in\!\Sigma$,
\begin{align}
    J_R\bigl(\beta_1,\!\beta_2,\sigma^*(\beta_1,\beta_2)\bigr)
    &\!\geq J_R(\beta_1,\beta_2,\sigma),
    \label{eq:eq_receiver}\\
    \!\!\!J_1\!\bigl(\!\beta_1^*(\beta_2),\!\beta_2,\!
    \sigma^*(\beta_1^*(\beta_2),\beta_2)\!\bigr)
    &\!\geq\! J_1\bigl(\!\beta_1,\beta_2,\sigma^*(\beta_1,\beta_2)\!\bigr),\!\!\!
    \label{eq:eq_follower}\\
    J_2\bigl(\beta_1^*(\beta_2^*),\beta_2^*,
    \sigma^*(\beta_1^*(\beta_2^*),&\beta_2^*)\bigr)
    \notag\\ &\hspace{-1.5cm}\geq J_2\bigl(\beta_1^*(\beta_2),\beta_2,
    \sigma^*(\beta_1^*(\beta_2),\beta_2)\bigr),
    \label{eq:eq_leader}
\end{align}
where ties in~\eqref{eq:eq_receiver} are broken as in
Section~\ref{subsec:utilities}.
\end{definition}

Under this tie-breaking rule, \eqref{eq:eq_receiver} yields the
response~\eqref{eq:receiver_response}. Since $J_1+J_2=1$, all
maximizers in~\eqref{eq:eq_follower} give the leader the same payoff,
so no tie-breaking rule is needed for the follower, and
\eqref{eq:eq_leader} states that $\beta_2^*$ minimizes the follower's
optimal payoff. With a slight abuse of notation, we also write
$(\beta_1^*,\beta_2^*,\sigma^*)$ for the equilibrium outcome, with
$\beta_1^*=\beta_1^*(\beta_2^*)$ and
$\sigma^*=\sigma^*(\beta_1^*,\beta_2^*)$.

\section{Equilibrium Characterization}
\label{sec:equilibrium}

We characterize the equilibria by backward induction. Taking the
receiver's response~\eqref{eq:receiver_response} as given, we first
derive sender~1's best response~\eqref{eq:eq_follower} to a fixed
leader policy and then determine sender~2's optimal
commitment~\eqref{eq:eq_leader}.

\subsection{Sender 1's Best Response}
\label{subsec:follower}
Fix $\beta_2=(s_2,c_2)\in\mathcal B_2$ and let
$\delta:=(c_{\min}-c_2)^+$, where $(z)^+:=\max\{z,0\}$. By
\eqref{eq:participation_threshold}, whenever $s_1+s_2>0$, the PC
holds if and only if $c_1\geq\delta$. Thus, $\delta$ is the minimum
state-$0$ intensity that sender~1 must supply for the receiver to
follow the reports.

\begin{lemma}
\label{lem:follower_br}
For every $\beta_2=(s_2,c_2)\in\mathcal B_2$, the policy
\begin{equation}
    \beta_1^*(\beta_2)=
    \begin{cases}
        (0,0), & R_1<\delta,\\
        (R_1-\delta,\delta), & R_1\geq\delta,
    \end{cases}
    \label{eq:follower_br_explicit}
\end{equation}
with $\delta=(c_{\min}-c_2)^+$, satisfies~\eqref{eq:eq_follower}.
If the follower's optimal payoff is positive, $\beta_1^*(\beta_2)$ is
the unique maximizer. If the optimal payoff is zero, every $\beta_1\in\mathcal B_1$ is a maximizer; we retain the selection in~\eqref{eq:follower_br_explicit} in the subsequent calculations.

\end{lemma}
\begin{proof}
For a given feasible leader policy $\beta_2=(s_2,c_2)$,
\eqref{eq:participation_threshold} implies that, whenever
$s_1+s_2>0$, the PC holds if and only if $c_1\geq\delta$. Thus,
$\delta$ is the minimum state-$0$ reporting intensity that sender~1
must supply for the receiver to follow the reports.

Suppose $R_1<\delta$. Then every $\beta_1\in\mathcal B_1$ satisfies
$c_1\leq R_1<\delta$. If $s_1+s_2>0$, the PC fails, so the receiver
selects $\sigma_d$ and sender~1 obtains a zero payoff. If
$s_1+s_2=0$, the estimate remains~$0$, and sender~1 again obtains a
zero payoff. Hence, every $\beta_1\in\mathcal B_1$, including
$(0,0)$, satisfies~\eqref{eq:eq_follower}.

Suppose $R_1\geq\delta$. Consider any $\beta_1\in\mathcal B_1$ with
$s_1+s_2>0$ that satisfies the PC. By the PC and the budget,
$c_1\geq\delta$ and $s_1\leq R_1-c_1\leq R_1-\delta$. Since
$s_1+s_2>0$ and $\delta+c_2=\max\{c_{\min},c_2\}>0$, \cite[Lemma 1]{GundoganBastopcu2026Timeliness} gives
\begin{equation}
\begin{aligned}
    J_1(s_1+s_2,c_1+c_2)
    &\leq J_1(s_1+s_2,\delta+c_2)\\
    &\leq J_1(R_1-\delta+s_2,\delta+c_2),
\end{aligned}
\label{eq:follower_dominance}
\end{equation}
where the first step lowers the state-$0$ intensity to its minimum
admissible value and the second raises the state-$1$ intensity to
its maximum feasible value. The policy $(R_1-\delta,\delta)$ is
feasible because $R_1\geq\delta$, satisfies the PC since
$\delta+c_2\geq c_{\min}$, and attains the right-hand side
of~\eqref{eq:follower_dominance}. Since policies that violate the PC
or have $s_1+s_2=0$ yield a zero payoff, $(R_1-\delta,\delta)$
satisfies~\eqref{eq:eq_follower}.

If the follower's optimal payoff is positive, every maximizer has
$s_1+s_2>0$ and satisfies the PC, so it obeys
\eqref{eq:follower_dominance}. Because the monotonicity
in \cite[Lemma 1]{GundoganBastopcu2026Timeliness} is strict, both inequalities
in~\eqref{eq:follower_dominance} must then hold with equality, which
gives $c_1=\delta$ and $s_1=R_1-\delta$ and hence the maximizer is
unique. If the optimal payoff is zero, every $\beta_1\in\mathcal B_1$
is a maximizer, since all payoffs are nonnegative.
\end{proof}

Lemma~\ref{lem:follower_br} shows how the leader's policy shapes the
follower's allocation. Since $\delta=(c_{\min}-c_2)^+$ decreases in
$c_2$ until it vanishes at $c_2=c_{\min}$, state-$0$ reports by
sender~2 reduce the intensity that sender~1 must devote to
participation and leave more of its budget for state-$1$ reports. At the boundary $R_1=\delta$, the selection in~\eqref{eq:follower_br_explicit} is $(0,\delta)$, so the follower's own state-$1$ intensity is zero. If $s_2>0$, this selection is the unique best response and yields a positive payoff, because the aggregate state-$0$ intensity $\delta+c_2\geq c_{\min}$ satisfies the PC and makes the leader's state-$1$ reports acceptable to the receiver.

\subsection{Sender 2's Optimal Policy}
\label{subsec:leader}

By~\eqref{eq:eq_leader} and $J_1+J_2=1$, sender~2 minimizes the
follower's optimal payoff. For $R_1>c_{\min}$ and $c_2\geq0$, let
$V(c_2)$ denote the follower's payoff when the leader uses
$\beta_2=(0,c_2)$. The next lemma shows that the leader need only
compare silence with full-budget state-$0$ reporting.

\begin{lemma}
\label{lem:leader_reduction}
There exists $\beta_2^*=(s_2^*,c_2^*)$ satisfying~\eqref{eq:eq_leader}
with $s_2^*=0$. If $R_1\leq c_{\min}$, the silent policy $(0,0)$
satisfies~\eqref{eq:eq_leader}. If $R_1>c_{\min}$, every $\beta_2^*$
satisfying~\eqref{eq:eq_leader} has $s_2^*=0$ and
$c_2^*\in\{0,R_2\}$.
\end{lemma}

\begin{proof}
Fix $\beta_2=(s_2,c_2)\in\mathcal B_2$ and let
$\delta=(c_{\min}-c_2)^+$. If $R_1<\delta$, the follower's optimal
payoff is zero for every $s_2$ by Lemma~\ref{lem:follower_br}.
Otherwise, the best response~\eqref{eq:follower_br_explicit} induces
the aggregate rates
\begin{equation}
    s=R_1+s_2-\delta,\qquad
    c=\max\{c_{\min},c_2\},
    \label{eq:leader_induced_rates}
\end{equation}
and, since $c>0$, the follower's optimal payoff $J_1(s,c)$ is
strictly increasing in $s_2$ by \cite[Lemma 1]{GundoganBastopcu2026Timeliness}. In both
cases, replacing $(s_2,c_2)$ with $(0,c_2)$ preserves feasibility
and does not increase the follower's optimal payoff, so it suffices
to consider $s_2=0$. If $R_1>c_{\min}$, then $R_1>\delta$ for every
$c_2$, and this replacement strictly decreases the follower's optimal
payoff whenever $s_2>0$, and thus every $\beta_2^*$
satisfying~\eqref{eq:eq_leader} has $s_2^*=0$.

If $R_1\leq c_{\min}$, the silent policy leaves the follower with a
zero optimal payoff by Lemma~\ref{lem:follower_br}, so it attains
$J_2=1$, which is the largest possible leader payoff. Now suppose
$R_1>c_{\min}$ and $s_2=0$. By~\eqref{eq:leader_induced_rates}, we have
\begin{equation}
    V(c_2)=
    \begin{cases}
        J_1(R_1\!-\!c_{\min}\!+\!c_2,c_{\min}),
            & 0\leq c_2\leq c_{\min},\\
        J_1(R_1,c_2), & c_2\geq c_{\min}.
    \end{cases}
    \label{eq:leader_value}
\end{equation}
As $c_2$ increases, only the state-$1$ intensity grows in the first
branch and only the state-$0$ intensity grows in the second. Hence,
by \cite[Lemma 1]{GundoganBastopcu2026Timeliness}, we obtain
\begin{equation}
\begin{aligned}
    V'(c_2)&>0, &&0<c_2<c_{\min},\\
    V'(c_2)&<0, &&c_2>c_{\min}.
\end{aligned}
    \label{eq:leader_value_derivative}
\end{equation}
Since the two branches agree at $c_2=c_{\min}$, $V(c_2)$ is continuous,
strictly increasing on $[0,c_{\min}]$, and strictly decreasing on
$[c_{\min},\infty)$. Its minimum over the feasible set $[0,R_2]$
therefore exists and is attained only at $c_2\in\{0,R_2\}$. Together
with the reduction to $s_2=0$, this proves all three claims.
\end{proof}

The proof of Lemma~\ref{lem:leader_reduction} identifies two
opposing effects of the leader's state-$0$ reports. For
$c_2<c_{\min}$, additional reports free part of the follower's budget
for state-$1$ reports while keeping the aggregate state-$0$ intensity
at $c_{\min}$, which benefits the follower. For $c_2\geq c_{\min}$,
the follower already devotes its entire budget to state-$1$ reports,
so additional state-$0$ reports benefit the leader. Hence, the leader strictly prefers full-budget reporting to silence
if and only if $V(R_2)<V(0)$.

Two budget thresholds, defined in Theorem~\ref{thm:equilibrium},
determine the equilibrium. The first, $R_1^{\mathrm{crit}}$, is the
follower budget above which intervention can be strictly preferred:
if $R_1\leq R_1^{\mathrm{crit}}$, silence is optimal for every $R_2$.
For $R_1>R_1^{\mathrm{crit}}$, the second, $R_2^{\mathrm{crit}}$, is
the leader budget above which intervention is strictly preferred,
with indifference at $R_2=R_2^{\mathrm{crit}}$.

\begin{theorem}
\label{thm:equilibrium}
Let $\phi(R_1):=R_1(R_1+\lambda+\mu-c_{\min})$ and
$\psi:=c_{\min}(c_{\min}+\lambda+\mu)$. Let $R_1^{\mathrm{crit}}$ be
the unique solution of $\phi(R_1)=\psi$ in $(c_{\min},\infty)$  given by
\begin{equation}
\begin{split}
    R_1^{\mathrm{crit}}
    =\frac{1}{2}&\Bigl(c_{\min}-\lambda-\mu\\
    &+\sqrt{(\lambda+\mu)^2+2c_{\min}(\lambda+\mu)+5c_{\min}^2}\Bigr),
\end{split}
\label{eq:R1_crit}
\end{equation}
and, for $R_1>c_{\min}$, let
\begin{equation}
    R_2^{\mathrm{crit}}:=
    \begin{cases}
        \dfrac{c_{\min}\,\phi(R_1)}{\phi(R_1)-\psi},
            & R_1>R_1^{\mathrm{crit}},\\[1.5ex]
        +\infty, & R_1\leq R_1^{\mathrm{crit}},
    \end{cases}
\label{eq:R2_crit}
\end{equation}
so that $R_2^{\mathrm{crit}}>c_{\min}$. The pure-strategy equilibrium
outcomes are as follows.
\begin{enumerate}
    \item If $R_1\leq c_{\min}$, the equilibrium outcomes are the
    triples $(\beta_1,\beta_2,\sigma^*(\beta_1,\beta_2))$ with
    $\beta_1\in\mathcal B_1$ and $\beta_2=(s_2,c_2)\in\mathcal B_2$
    such that either $R_1<\delta$, or $R_1=\delta$ and $s_2=0$,
    where $\delta:=(c_{\min}-c_2)^+$. All of them yield $J_1=0$ and
    $J_2=1$; in particular, $((0,0),(0,0),\sigma_d)$ is an
    equilibrium outcome.
    \item If $R_1>c_{\min}$, we define the silence and intervention
    outcomes as
    \begin{align*}
        \mathcal O_{\mathrm s}
        &:=\bigl((R_1-c_{\min},c_{\min}),(0,0),\sigma_s\bigr),\\
        \mathcal O_{\mathrm i}
        &:=\bigl((R_1,0),(0,R_2),\sigma_s\bigr).
    \end{align*}
    Then $\mathcal O_{\mathrm s}$ is the unique equilibrium outcome
    if $R_2<R_2^{\mathrm{crit}}$, $\mathcal O_{\mathrm i}$ is the
    unique equilibrium outcome if $R_2>R_2^{\mathrm{crit}}$, and
    these two are the only equilibrium outcomes if
    $R_2=R_2^{\mathrm{crit}}$.
\end{enumerate}
\end{theorem}
\begin{proof}
\emph{Part~1.} Suppose $R_1\leq c_{\min}$, and fix
$\beta_2=(s_2,c_2)\in\mathcal B_2$ with $\delta=(c_{\min}-c_2)^+$.
By Lemma~\ref{lem:follower_br}, the follower's optimal payoff is zero
if $R_1<\delta$. If $R_1\geq\delta$, it equals
$J_1(R_1-\delta+s_2,\max\{c_{\min},c_2\})$, which
by~\eqref{eq:sender_average} vanishes if and only if
$R_1-\delta+s_2=0$, that is, $R_1=\delta$ and $s_2=0$. Silence
satisfies this condition, so the minimum follower payoff is zero, and
$\beta_2$ satisfies~\eqref{eq:eq_leader} if and only if the stated
condition holds. For each such $\beta_2$, every
$\beta_1\in\mathcal B_1$ is a maximizer by
Lemma~\ref{lem:follower_br} and yields $J_1=0$ and $J_2=1$. With
$\beta_1=\beta_2=(0,0)$, the receiver selects $\sigma_d$
by~\eqref{eq:receiver_response}.

\emph{Part~2.} Suppose $R_1>c_{\min}$. By
Lemma~\ref{lem:leader_reduction}, the leader's candidates are
$(0,0)$ and $(0,R_2)$. If $R_2\leq c_{\min}$, then $V(R_2)>V(0)$
unless $R_2=0$, because $V$ is strictly increasing on
$[0,c_{\min}]$, silence is the unique optimal policy,
consistent with $R_2<R_2^{\mathrm{crit}}$. Now let $R_2>c_{\min}$.
By \eqref{eq:sender_average} and \eqref{eq:leader_value}, we have
\begin{align*}
    V(0)&=J_1(R_1-c_{\min},c_{\min})
    =\frac{\lambda N_0}{(\lambda+\mu)D_0},\\
    V(R_2)&=J_1(R_1,R_2)
    =\frac{\lambda N_2}{(\lambda+\mu)D_2},
\end{align*}
where $N_0:=(R_1-c_{\min})(c_{\min}+\lambda+\mu)$,
$D_0:=\mu c_{\min}+(\lambda+c_{\min})(R_1-c_{\min})$,
$N_2:=R_1(R_2+\lambda+\mu)$, and $D_2:=\mu R_2+(\lambda+R_2)R_1$ are
positive. Direct expansion gives
\begin{equation}
    N_0D_2-N_2D_0
    =\mu\bigl[(\phi(R_1)-\psi)R_2-c_{\min}\phi(R_1)\bigr],
    \label{eq:leader_comparison}
\end{equation}
so $V(R_2)$ is smaller than, equal to, or larger than $V(0)$
according as the bracket is positive, zero, or negative. The function
$\phi-\psi$ equals $-c_{\min}^2$ at $c_{\min}$ and has derivative
$2R_1+\lambda+\mu-c_{\min}>0$ for $R_1\geq c_{\min}$. Hence, it has
the unique root $R_1^{\mathrm{crit}}$ in $(c_{\min},\infty)$, given
by~\eqref{eq:R1_crit}, and is positive exactly when
$R_1>R_1^{\mathrm{crit}}$. If $R_1\leq R_1^{\mathrm{crit}}$, the
bracket is negative for every $R_2$, so silence is strictly optimal.
If $R_1>R_1^{\mathrm{crit}}$, the bracket has the sign of
$R_2-R_2^{\mathrm{crit}}$, where $R_2^{\mathrm{crit}}>c_{\min}$
because $\psi>0$. Thus, silence is the unique optimal policy if
$R_2<R_2^{\mathrm{crit}}$, full-budget state-0 reporting is the unique optimal
policy if $R_2>R_2^{\mathrm{crit}}$, and both are optimal if
$R_2=R_2^{\mathrm{crit}}$.

Finally, under silence, $\delta=c_{\min}<R_1$ and the follower's
optimal payoff is positive, so its unique best response is
$(R_1-c_{\min},c_{\min})$. Under optimal intervention,
$R_2\geq R_2^{\mathrm{crit}}>c_{\min}$, so $\delta=0$ and its unique
best response is $(R_1,0)$. In both cases, $s>0$ and $c\geq c_{\min}$,
so the receiver selects $\sigma_s$ by~\eqref{eq:receiver_response}.
This yields $\mathcal O_{\mathrm s}$ and $\mathcal O_{\mathrm i}$.
\end{proof}

Fig.~\ref{fig:sender2_utilities} illustrates the leader's choice
between silence and intervention. At $R_2=R_2^{\mathrm{crit}}$,
$\mathcal O_{\mathrm s}$ and $\mathcal O_{\mathrm i}$ yield identical
sender payoffs but, as Section~\ref{sec:receiver_benefit} shows,
different receiver utilities.

\begin{figure}[t]
  \centering
  \begin{subfigure}{0.48\columnwidth}
    \centering
    \includegraphics[width=\linewidth]{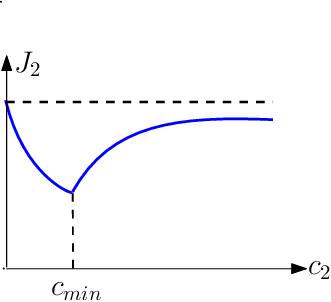}
    \caption{}
  \end{subfigure}
  \hfill
  \begin{subfigure}{0.48\columnwidth}
    \centering
    \includegraphics[width=\linewidth]{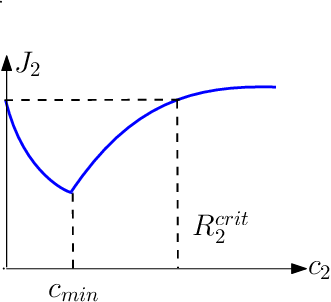}
    \caption{}
  \end{subfigure}
  \caption{Schematic illustration of sender~2's payoff $1-V(c_2)$ under $\beta_2=(0,c_2)$, given
the follower's best response, for
(a) $c_{\min}<R_1<R_1^{\mathrm{crit}}$ and
(b) $R_1>R_1^{\mathrm{crit}}$. The payoff decreases on
$[0,c_{\min}]$ and increases thereafter. The dashed line marks the
silence payoff $1-V(0)$. In (a), every $c_2>0$ yields less than
silence. In (b), the payoff exceeds that of silence exactly when
$c_2>R_2^{\mathrm{crit}}$, so full-budget reporting is optimal if
$R_2>R_2^{\mathrm{crit}}$.}
  \label{fig:sender2_utilities}
  \vspace{-0.5cm}
\end{figure}
\section{Receiver Benefit from Competition}
\label{sec:receiver_benefit}

As a benchmark, we consider a single sender with sender~1's preferences
and total budget $R:=R_1+R_2$, with the source, receiver, and budget
model unchanged. This benchmark coincides with our game for
$(R_1,R_2)=(R,0)$, in which the leader can only remain silent. By
\cite[Theorem~1]{GundoganBastopcu2026Timeliness}, or by applying
Theorem~\ref{thm:equilibrium} with $R_2=0$, the receiver's utility
under this benchmark is
\begin{equation}
    J_R^{\mathrm{single}}(R)
    =q\pi_0=\frac{q\mu}{\lambda+\mu},\qquad R\geq0.
    \label{eq:single_sender_benchmark}
\end{equation}
If the single sender instead had sender~2's preferences, it could
not gain from reporting, because the receiver's default estimate~$0$
is already the one it prefers. An optimal policy for such a sender is therefore to remain silent,
in which case the estimate remains~$0$ and the receiver again obtains
$q\pi_0$.
For the benchmark sender with sender~1's preferences, if
$R\leq c_{\min}$, the sender cannot obtain a positive payoff while
satisfying the PC, whereas for $R>c_{\min}$ it chooses
$(s,c)=(R-c_{\min},c_{\min})$, which makes the PC bind. In both
cases, the receiver obtains only its default utility, even though a
single sender controls the entire budget. More generally,
by~\eqref{eq:receiver_gain_identity}, ZOH strictly improves on the
default utility exactly when $s>0$ and $c>c_{\min}$, and the gain is
zero when $s=c=0$. Competition can therefore benefit the receiver
only through an equilibrium with $s>0$ and $c>c_{\min}$.

\begin{theorem}
\label{thm:receiver_benefit}
Let $J_R^*$ denote the receiver's utility at an equilibrium outcome
of Theorem~\ref{thm:equilibrium}. For all $R_1,R_2\geq0$,
\begin{equation}
    J_R^*\geq J_R^{\mathrm{single}}(R_1+R_2)=q\pi_0,
    \label{eq:receiver_weak_benefit}
\end{equation}
with equality at every equilibrium outcome other than
$\mathcal O_{\mathrm i}$. The outcome $\mathcal O_{\mathrm i}$ is an
equilibrium outcome if and only if $R_1>c_{\min}$ and
$R_2\geq R_2^{\mathrm{crit}}$, in which case
\begin{equation}
    J_R^*-q\pi_0
    =\frac{(1-q)\lambda R_1(R_2-c_{\min})}
    {(\lambda+\mu)(\mu R_2+\lambda R_1+R_1R_2)}
    >0.
    \label{eq:receiver_equilibrium_gain}
\end{equation}
In particular, at $R_2=R_2^{\mathrm{crit}}$, the two equilibrium
outcomes give the senders identical payoffs but give the receiver
$q\pi_0$ under $\!\mathcal O_{\mathrm s}\!$ and
$\!q\pi_0\!$ plus the strictly positive gain in
\eqref{eq:receiver_equilibrium_gain} under $\!\mathcal O_{\mathrm i}$.
\end{theorem}

\begin{proof}
By~\eqref{eq:eq_receiver}, the receiver's equilibrium utility is at
least $q\pi_0$, the utility it obtains under $\sigma_d$. Together
with~\eqref{eq:single_sender_benchmark}, this
proves~\eqref{eq:receiver_weak_benefit}.

By Theorem~\ref{thm:equilibrium}, every equilibrium outcome other
than $\mathcal O_{\mathrm i}$ either has $R_1\leq c_{\min}$ or equals
$\mathcal O_{\mathrm s}$. In the first case, $J_1=0$ in equilibrium,
whereas~\eqref{eq:sender_average} gives $J_1(s,c)>0$ whenever the
receiver follows reports with $s>0$. Thus, the receiver either
selects $\sigma_d$ or follows reports with $s=0$, and in both cases
$J_R^*=q\pi_0$. Under $\mathcal O_{\mathrm s}$, the aggregate rates
are $(s,c)=(R_1-c_{\min},c_{\min})$,
so~\eqref{eq:receiver_gain_identity} gives $J_R^*=q\pi_0$.

By Theorem~\ref{thm:equilibrium}, $\mathcal O_{\mathrm i}$ is an
equilibrium outcome if and only if $R_1>c_{\min}$ and
$R_2\geq R_2^{\mathrm{crit}}$. Its aggregate rates are
$(s,c)=(R_1,R_2)$, so~\eqref{eq:receiver_gain_identity}
gives~\eqref{eq:receiver_equilibrium_gain}, which is positive because
$R_2\!\geq \!R_2^{\mathrm{crit}}\!>\!c_{\min}$. Finally, at
$R_2\!=\!R_2^{\mathrm{crit}}$, both outcomes are equilibria. They give the senders identical payoffs, since the leader is indifferent
between them, but give the receiver $q\pi_0$ under
$\mathcal O_{\mathrm s}$ and $q\pi_0$ plus the strictly positive gain
in~\eqref{eq:receiver_equilibrium_gain} under $\mathcal O_{\mathrm i}$.
\end{proof}

The gain arises because the opposing sender makes the PC slack. A
single sender supplies only the state-$0$ intensity $c_{\min}$ that
participation requires, so the PC binds. Under
$\mathcal O_{\mathrm i}$, sender~2 supplies $c_2=R_2>c_{\min}$, while
sender~1 already sets $c_1=0$ and therefore cannot offset these
reports by reducing its own state-$0$ intensity. Hence, competition
raises the receiver's utility even though the total budget is
unchanged.

For fixed $R_1>R_1^{\mathrm{crit}}$, the equilibrium receiver-utility
correspondence has distinct branches on either side of
$R_2^{\mathrm{crit}}$. For $R_2<R_2^{\mathrm{crit}}$, the unique
equilibrium yields $q\pi_0$, whereas for
$R_2>R_2^{\mathrm{crit}}$, the unique intervention equilibrium yields
the strictly larger utility in~\eqref{eq:receiver_equilibrium_gain}.
At $R_2=R_2^{\mathrm{crit}}$, both values are equilibrium utilities.
The sender payoffs, in contrast, are continuous across the threshold
because $\mathcal O_{\mathrm s}$ and $\mathcal O_{\mathrm i}$ give
them identical payoffs there.

\section{Numerical Results}
\label{sec:numerical}

We consider a binary source with $\lambda=2$, $\mu=3$, and $q=0.5$,
which satisfy $q\pi_0>(1-q)\pi_1$. These parameters give
$c_{\min}=1$, default receiver utility $q\pi_0=0.3$, and
$R_1^{\mathrm{crit}}=\sqrt{10}-2\approx1.1623$. All equilibrium
curves follow from Theorem~\ref{thm:equilibrium}, and the receiver's
utilities from Theorem~\ref{thm:receiver_benefit}.

We first fix $R_1=2$ and vary $R_2$, so that $\phi(2)=12$, $\psi=6$,
and $R_2^{\mathrm{crit}}=2$. As Fig.~\ref{fig:numres1} shows, the
receiver obtains $q\pi_0=0.3$ for $R_2<2$, where
$\mathcal O_{\mathrm s}$ is the unique equilibrium outcome. At
$R_2=2$, both $\mathcal O_{\mathrm s}$ and $\mathcal O_{\mathrm i}$
are equilibrium outcomes, giving the receiver $0.3$ and
$0.3286$, respectively. For $R_2>2$, the receiver's
utility under $\mathcal O_{\mathrm i}$ increases with $R_2$ and
approaches $0.38$ as $R_2$ grows, illustrating the strict gain in
Theorem~\ref{thm:receiver_benefit}.

Fig.~\ref{fig:numres2} shows the corresponding sender payoffs. For
$R_2\leq2$, they equal $J_1^*=V(0)=0.4$ and $J_2^*=0.6$. For $R_2>2$,
the follower's payoff $V(R_2)$ decreases and approaches $0.16$ as
$R_2$ grows, while the leader's payoff increases toward $0.84$. The
sender payoffs are thus continuous at $R_2=2$, whereas the receiver's
utility is not.

\begin{figure}[t]
    \centering
    \includegraphics[width=0.85\columnwidth]{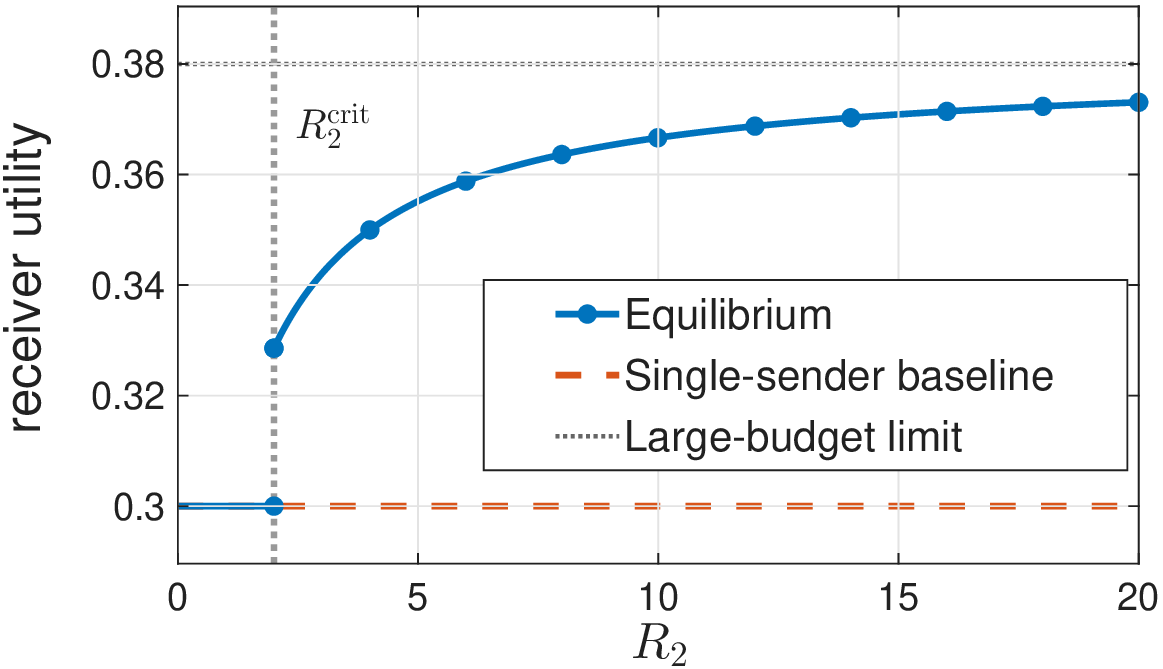}
    \caption{Equilibrium receiver utility $J_R^*$ versus $R_2$ for
$R_1\!=\!2$, where $R_2^{\mathrm{crit}}\!=\!2$. The two marked points at
$R_2=2$ are the utilities under $\mathcal O_{\mathrm s}$ and
$\mathcal O_{\mathrm i}$. The horizontal lines show the single-sender
benchmark $q\pi_0\!=\!0.3$ and the limit $0.38$ as $R_2$ grows.}
    \label{fig:numres1}
    \vspace{-0.4cm}
\end{figure}

\begin{figure}[t]
    \centering
    \includegraphics[width=0.85\columnwidth]{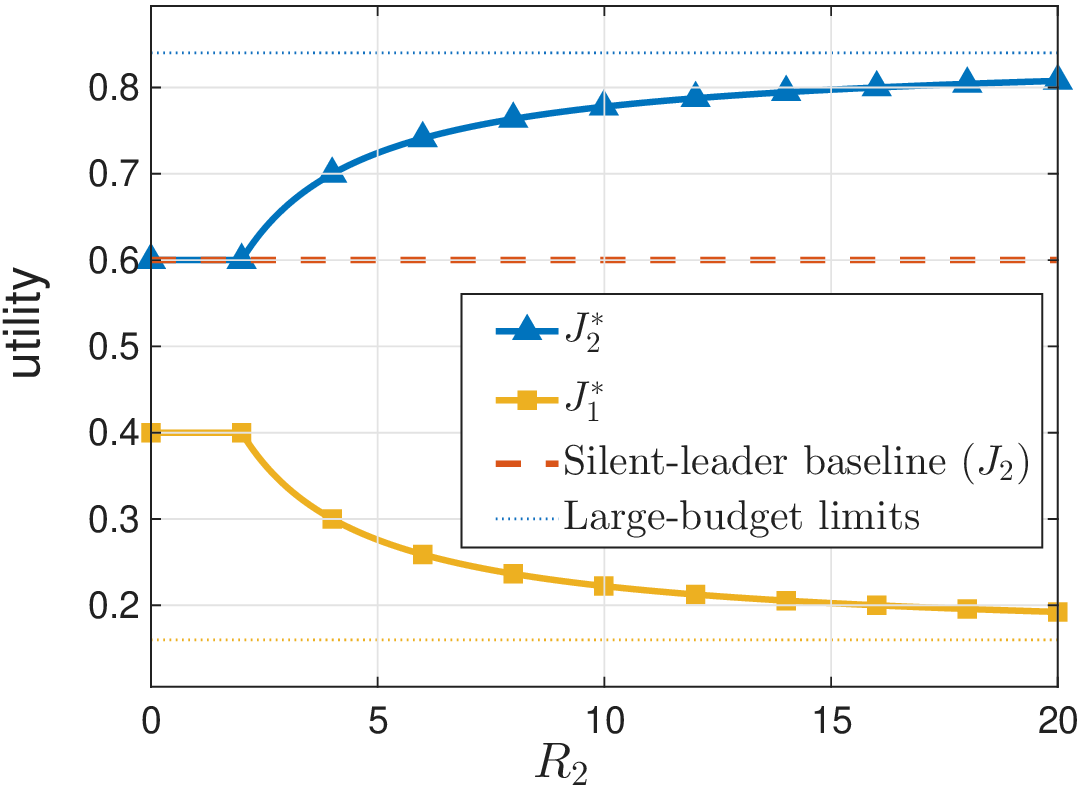}
    \caption{Equilibrium sender payoffs $J_1^*$ and $J_2^*$ versus $R_2$
for $R_1=2$, where $R_2^{\mathrm{crit}}=2$.}
    \label{fig:numres2}
    \vspace{-0.6cm}
\end{figure}

Finally, Fig.~\ref{fig:numres3} fixes $R_1+R_2=8$ and varies the
budget division. We compare the equilibrium with a receiver-optimal
benchmark that chooses the aggregate intensities directly, under the
same budget and the same choice between ZOH and the default
estimator:
\begin{equation}
    J_R^{\mathrm{opt}}(8)
    =\max_{\substack{s,c\geq0\\s+c\leq8}}
    \max\{q\pi_0,J_R(s,c)\}.
    \label{eq:receiver_optimal_benchmark}
\end{equation}
By~\eqref{eq:receiver_gain_identity}, $J_R(s,c)$ is strictly
increasing in $c$ whenever $s>0$, since its $c$-derivative has the
sign of $\lambda s+c_{\min}(\mu+s)>0$. Hence, any maximizer above the
default utility satisfies $s+c=8$. Substituting $c=8-s$ gives
\[
    J_R(s,8-s)
    =0.3+\frac{s(7-s)}{5[24+s(7-s)]},
\]
which is increasing in $s(7-s)$ and hence maximized at $s=3.5$. The
benchmark is therefore attained at $(s,c)=(3.5,4.5)$, with
$J_R^{\mathrm{opt}}(8)\approx0.3676$.

At $(R_1,R_2)=(3.5,4.5)$, we have $R_2>R_2^{\mathrm{crit}}\approx1.30$,
so $\mathcal O_{\mathrm i}$ is the unique equilibrium outcome:
sender~1 reports only state~$1$, sender~2 reports only state~$0$, and
the aggregate intensities $(s,c)=(3.5,4.5)$ coincide with the
benchmark maximizer. More generally, since $\mathcal O_{\mathrm i}$ implements
$(s,c)=(R_1,R_2)$, any benchmark maximizer $(s,c)$ satisfying
$s>R_1^{\mathrm{crit}}$ and
$c>R_2^{\mathrm{crit}}(s)$ is uniquely implemented by the
corresponding budget division $(R_1,R_2)=(s,c)$, where
$R_2^{\mathrm{crit}}(s)$ denotes~\eqref{eq:R2_crit} evaluated at
$R_1=s$. At $c=R_2^{\mathrm{crit}}(s)$, the same reporting rates are
attained under the intervention equilibrium, but the equilibrium
outcome is not unique. Moving away from $(3.5,4.5)$ reduces the receiver's
utility, and if either sender receives too little budget, the receiver obtains only its default utility $0.3$, as under a single sender with the entire
budget. Thus, in this example, an appropriate budget division allows
strategic reporting to attain receiver-optimal performance within the
specified estimator architecture.

\begin{figure}[t]
    \centering
    \includegraphics[width=0.95\columnwidth]{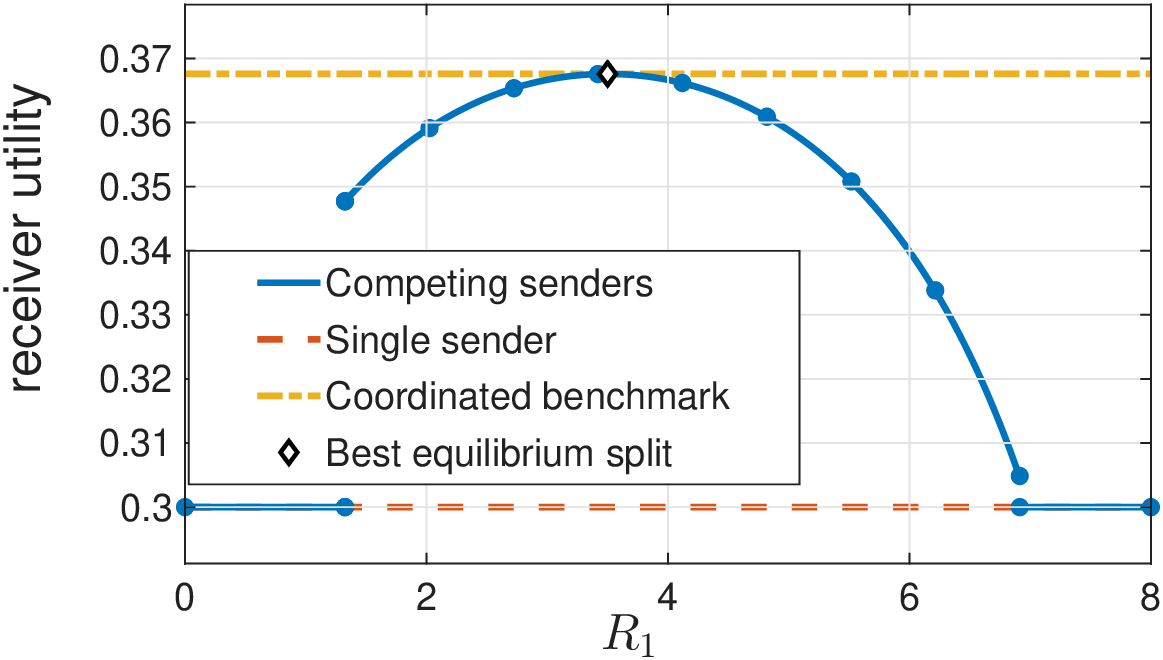}
    \caption{Equilibrium receiver utility $J_R^*$ versus the budget
division $(R_1,8-R_1)$, compared with the receiver-optimal benchmark
$J_R^{\mathrm{opt}}(8)\approx0.3676$.}
    \label{fig:numres3}
    \vspace{-0.5cm}
\end{figure}

\section{Conclusion}
We studied competition between two senders that influence a receiver through the timing of truthful reports about a binary continuous-time Markov source. We characterized the Stackelberg equilibria in closed form. When the follower can obtain positive utility against a silent leader, the leader's state-$0$ reports initially help the follower by reducing the reporting effort needed for participation, so every optimal leader policy is either silence or full-budget state-$0$ reporting, with explicit budget thresholds determining which. In the low-budget regime, equilibrium leader policies can be nonunique, but all such equilibria leave the receiver at its default utility. When the full-budget intervention outcome is optimal, it makes the participation constraint slack and strictly improves the receiver's utility relative to a single sender with the same total budget. In our numerical example, an appropriate budget division attains a receiver-optimal reporting benchmark. Future work will consider opposing senders that allocate their budgets across multiple heterogeneous sources, alternative commitment orders, and receivers that also exploit information conveyed by the absence of reports.

\bibliographystyle{ieeetr}
\bibliography{references2}
\end{document}